\documentclass[runningheads,envcountsame, a4paper]{llncs}
\usepackage{amssymb}
\usepackage{epstopdf}
\usepackage{graphicx}
\usepackage[hyphens]{url}
\usepackage{microtype}
\usepackage{url}
\urldef{\mailsa}\path|{alfred.hofmann, ursula.barth, ingrid.haas, frank.holzwarth,|
\urldef{\mailsb}\path|anna.kramer, leonie.kunz, christine.reiss, nicole.sator,|
\urldef{\mailsc}\path|erika.siebert-cole, peter.strasser, lncs}@springer.com|    
\newcommand{\keywords}[1]{\par\addvspace\baselineskip
\noindent\keywordname\enspace\ignorespaces#1}

\usepackage{color}

\def\P{{\mathchoice {\hbox{$\sf\textstyle P\kern-0.4em Z$}}
{\hbox{$\sf\textstyle P\kern-0.4em P$}}
{\hbox{$\sf\scriptstyle P\kern-0.3em P$}}
{\hbox{$\sf\scriptscriptstyle P\kern-0.2em P$}}}}
\def\cqfd{\par\nopagebreak\rightline{\vrule height 3pt width 5pt depth 2pt}
\medbreak}
\usepackage{epsf}
\let\Rightarrow=\Rightarrow

\def\rem#1\par{\par\noindent\begin{rema} \nopagebreak \strut \rm #1 \end{rema}}

\def\ex#1\par{\par\noindent\begin{exemple} \nopagebreak \strut \rm #1 \end{exemple}}

\def\pb#1\par{\par\noindent\begin{pb} \nopagebreak \strut #1 \end{pb}}

\def\skpr{\mbox{\it Sketch proof. \/}}

\def\N{{\mathchoice {\hbox{$\sf\textstyle N\kern-0.4em N$}}
{\hbox{$\sf\textstyle I\kern-0.42em N$}}
{\hbox{$\sf\scriptstyle I\kern-0.2em N$}} 
{\hbox{$\sf\scriptscriptstyle I\kern-0.em N$}}}}
\def\Z{{\mathchoice {\hbox{$\sf\textstyle Z\kern-0.4em Z$}}
{\hbox{$\sf\textstyle Z\kern-0.4em Z$}}
{\hbox{$\sf\scriptstyle Z\kern-0.3em Z$}}
{\hbox{$\sf\scriptscriptstyle Z\kern-0.2em Z$}}}}

\author{Jean N\'eraud \and Carla Selmi}
\institute{Universit\'e de Rouen, Laboratoire d'Informatique, de Traitement de l'Information et des Syst\`emes, 
 Avenue de l'Universit\'e, 76800 Saint-\'Etienne-du-Rouvray, France\\
\email{neraud.jean@free.fr,carla.selmi@univ-rouen.fr}  
}

\begin{document}
\title{Invariance: a Theoretical Approach for Coding Sets of Words Modulo Literal (Anti)Morphisms}
\titlerunning{A Theoretical  Approach for Coding Sets of Words Modulo Literal (Anti)Morphisms }
\toctitle{Invariance: a Theoretical Approach for Coding Sets of Words Modulo Literal (Anti)Morphisms}
\tocauthor{J.~N\'eraud and C.~Selmi}
\authorrunning{J. N\'eraud and C. Selmi}
\maketitle
\setcounter{footnote}{0}
\begin{abstract}
Let $A$ be a finite or countable alphabet and let $\theta$ be  literal (anti)morphism onto $A^*$ (by definition, such a correspondence is determinated by a permutation of the alphabet). This paper deals with sets which are invariant under $\theta$ ($\theta$-invariant for short).
We establish an extension of the famous defect theorem. Moreover, we prove that for the 
so-called thin $\theta$-invariant codes, maximality and completeness are two equivalent notions. 
We prove that a similar property holds for some special families of $\theta$-invariant codes such as prefix (bifix) codes, codes with a finite (two-way) deciphering delay, uniformly synchronous codes and circular codes. For a  special class of involutive antimorphisms, we prove that any regular $\theta$-invariant code may be embedded into a complete one.
\end{abstract}

\section{Introduction}
During the last decade, in the free monoid theory, due to their powerful applications, in particular in DNA-computing, one-to-one {\it morphic} or {\it antimorphic} correspondences play a particularly important part.
Given a finite or countable {\it alphabet}, say $A$, any  such  mapping  is  a substitution which is  fully determined by extending a unique  permutation of $A$, to a mapping onto $A^*$  (the {\it free monoid} that is generated by $A$).
The resulting mapping is commonly referred to as  {\it literal} (or {\it  letter-to-letter}) moreover, in the case of a finite alphabet, 
it is well known that,  with respect to the composition, some power of such a correspondence is the identity
(classically, in the case where this power corresponds to the square, we say that the correspondence is {\it involutive}).

 In that special case of involutive morphisms or antimorphisms -we write (anti)morphisms for short,  lots of successful investigations have been done 
for extending the now classical  combinatorical properties on  words:
we mention the study of the so-called pseudo-palindromes \cite{BdZD11,dD06}, or that of pseudo-repetitions \cite{ARVS14,GMRMNT13,KM08}.
The framework of  some peculiar families of codes \cite{KM06} and equations in words \cite{CCKS11,MMNS14} have been also concerned.
Moreover, in the larger family  of  one-to-one (anti)morphisms, a nice generalization of 
the famous theorem of Fine and Wilf \cite[Proposition 1.3.5]{Lo83} has been recently established in \cite{MMN11}.

Equations in words are also the starting point of the study in the present paper, where we adopt the point of view from \cite[Chap. 9]{Lo83}.
Let $A$ be a finite or countable alphabet; a one-to-one literal  (anti)morphism onto $A^*$, namely $\theta$,  being fixed, consider a finite collection of unknown words, say $Z$. 
In view of making the present foreword more readable,  in the first instance we take $\theta$  as an involutive literal substitution (that is $\theta^2=id_{A^*}$).
We assign that the words in $Z$ and their images by $\theta$ to satisfy a given equation, and we are interested in the cardinality of any set $T$, whose elements allow by concatenation to compute all the words in $Z$.
Actually, such a question  might be more complex than in the classical configuration, where $\theta$ does not interfer:
it is well known that in that classical case, according to the famous defect theorem \cite[Theorem 1.2.5]{Lo83}, 
the words in $Z$ may be computed as the concatenation of at most $|Z|-1$ words that  don't satisfy any non-trivial equation.
With the terminology of \cite{Lo83,BPR10}, $T$, the set of such words is a {\it code}, or equivalently  $T^*$, the submonoid that it generates, is {\it free}: more precisely, with respect to the inclusion of sets it is the smallest free submonoid of $A^*$ that contains $Z$.

Along the way, for solving our problem, applying the defect theorem  to  the set $X=Z\cup\theta(Z)$  might appear as natural.  Such a methodology garantees 
the existence of a code $T$, with $|T|\le |X|-1$, and such that $T^*$ is the smallest free submonoid of $A^*$ that contains $X$.
Unfortunately, since both the words in $Z$ and $\theta(Z)$ are expressed  as concatenations of 
words in $T$,  among the elements of 
$T\cup\theta(T)$ non-trivial equations can remain; in other words, by applying that methodology, 
the initial problem would be transferred among the words in $T\cup\theta(T)$. This situation is particularly illustrated by \cite[Proposition 3]{KM08}, where the authors prove that, given an involutive antimorphism $\theta$, the solutions of  the equation $xy=\theta(y)x$ are $x=(uv)^i u, y=vu$, where the elements $u,v$ of $T$ satisfy the non-trivial equation  $vu=\theta(u)\theta(v)$.

In the general case where $\theta$ is a literal one-to-one (anti)morphism, we note that
the union, say $Y$,  of the sets $\theta^i(T)$, for all $i\in \Z$,  is itself $\theta$-invariant, therefore
an alternative methodology will consist in asking for some code $Y$ which is invariant under $\theta$, and such that $Y^*$ is the smallest free submonoid of $A^*$ that contains $X=\bigcup_{i\in \Z} \theta^i(Z)$.
By the way, it is straightforward to prove that the intersection of an arbitrary family of  $\theta$-invariant free  submonoids is itself a
$\theta$-invariant free submonoid. In the present  paper we prove the following result:
\begin{flushleft}
{\bf Theorem 1.}
{\it Let $A$ be a finite or countable alphabet, let  $\theta$ be a literal (anti)morphism onto $A^*$, and let $X$ be a  finite $\theta$-invariant set.
If $X$  it  is not a code, then the smallest  $\theta$-invariant free submonoid of $A^*$ that contains $X$ 
is generated by a $\theta$-invariant code $Y$ which satisfies  $|Y|\le |X|-1$.
}
\end{flushleft}
For illustrating this result in term of equations, we refer to \cite{CCKS11,MMNS14}, where
the authors  considered generalizations of the famous equation in three unknowns of Lyndon-Sh\"utzenberger  \cite[Sect. 9.2]{Lo83}. 
They proved that, an involutive (anti)morphism $\theta$ being fixed, given such an  equation with sufficiently long members, 
a word $t$ exists such that any 3-uple of ``solutions" can be expressed as a concatenation of words in $\{t\}\cup \{\theta(t)\}$.
With the notation of Theorem 1, the elements of the $\theta$-invariant set $X$ are $x,y,z,\theta(x),\theta(y),\theta(z)$ 
and those of $Y$ are $t$ and $\theta(t)$: we verify that $Y$ is a $\theta$-invariant code with $|Y|\le |X|-1$.

In the sequel, we will  continue our investigation by studying the properties of  complete $\theta$-invariant codes:
a subset $X$ of $A^*$ is {\it complete} if  any word of $A^*$ is a factor of some words in $X^*$. 
From this point of view, a famous result from Sch\"utzenberger states that, for the wide family of the so-called {\it thin} codes 
(which contains regular codes) \cite[Sect. 2.5]{BPR10}, maximality and completeness are two equivalent notions.
In the framework of invariant codes, we prove the following result:
\begin{flushleft}
{\bf Theorem 2.}
{\it Let $A$ be a finite or countable alphabet. Given a thin $\theta$-invariant code $X\subseteq A^*$,  the three following conditions are equivalent:\\
(i) $X$ is complete\\
(ii) $X$ is a maximal code\\
(iii) $X$ is maximal in the family of the $\theta$-invariant codes.
}
\end{flushleft}
In the proof,  the main feature consists in establishing that a non-complete $\theta$-invariant code $X$ 
cannot be maximal in the family of $\theta$-invariant codes: 
actually, the most delicate step lays upon the construction of a convenient $\theta$-invariant set 
$Z\subseteq A^*$, with $X\cap Z=\emptyset$ and 
such that $X\cup Z$ remains itself a $\theta$-invariant code.

It is well known that the preceding result from Sch\"utzenberger  has been successfully  extended to some famous families of thin codes, 
such as {\it prefix} ({\it bifix, uniformly synchronous, circular}) codes 
(cf \cite[Proposition 3.3.8]{BPR10}, \cite[Proposition 6.2.1]{BPR10}, \cite[Theorem 10.2.11]{BPR10},  \cite[Proposition 3.6]{Br98} and  \cite[Theorem 3.5]{N08}) and codes with a {\it finite deciphering delay} (f.d.d. codes, for short)  \cite[Theorem 5.2.2]{BPR10}.
From this point of view, we will examine the behavior of corresponding  families of $\theta$-invariant codes.
Actually we establish a result similar to the preceding theorem 2 in the framework of the family of prefix
(bifix, f.d.d., two-way f.d.d, uniformly synchronized, circular codes).
In the proof,  a construction very similar to the previous one may be  used in the case of prefix, bifix, f.d.d., two-way f.d.d codes. 
At the contrary, investigating the behavior of circular codes  with regards to the question necessitates the computation of a more sofisticated set; moreover
the family of uniformly synchronized codes itself impose to make use of a significantly different methodology.

In the last part of our study, we address  to the problem of embedding a non-complete $\theta$-invariant code into  a complete one. 
For the first time, this question was  stated  in \cite{R75}, where the author asked whether any finite code can be imbedded into a regular one. 
A positive answer was provided in \cite{ER85}, where was established a formula for embedding any regular code into a complete one.
From the point of view of $\theta$-invariant codes, we obtain a positive answer only in the case 
where $\theta$ is an involutive antimorphism which is different of the so-called miror image;
actually the general question remains open.

We now describe the contents of the paper. 
Section 2 contains the preliminaries: the terminology of the free monoid is settled, 
and the definitions of some classical families of codes are recalled. Theorem 1 is established in Section 3, 
where an original example of equation  is studied.
The proof of Theorem 2 is done in Section 3, and extensions for special familes of $\theta$-invariant codes are studied in Section 4.
The question of embedding a regular $\theta$-invariant code into a complete one is examined in Section 5.
\section{Preliminaries}
We adopt the notation of the free monoid theory: given an alphabet $A$, we denote by $A^*$ the free monoid that it generates. 
Given a word $w$, we denote by $|w|$ its length, the empty word, that we denote by $\varepsilon$, being the word with length $0$. We denote by $w_i$ the letter of position $i$ in $w$: with this notation we have $w=w_1\cdots w_{|w|}$.
We set $A^+=A^*\setminus \{\varepsilon\}$. 
Given $x\in A^*$ and $w\in A^+$, we say that $x$ is a {\it prefix} ({\it suffix}) of $w$ if a word $u$ exists such that $w=xu$ ($w=ux$).
Similarly, $x$ is a {\it factor} of $w$ if a pair of words $u,v$ exists such that $w=uxv$. 
Given a non-empty set $X\subseteq A^*$, we denote by $P(X)$ ($S(X), F(X)$) the set of the words that are prefix (suffix, factor) of some word in $X$. Clearly, we have $X\subseteq P(X)$  ($S(X), F(X)$). A set $X\subseteq A^*$ is {\it complete} iff $F(X^*)=A^*$.
Given a pair of words $w,w'$, we say that it {\it overlaps} if  words $u,v$ exist such that $uw'=wv$ or $w'u=vw$,
with $1\le |u| < |w|$ and  $1\le |v| < |w'|$;  
otherwise, the pair is {\it overlapping-free} (in such a case, if $w=w'$, we simply say that $w$ is overlapping-free).

It is assumed that the reader has a fundamental understanding  with the main concepts of the theory of variable length codes: we only recall some of the main definitions and we
suggest, if necessary,  that he (she) report to  \cite{BPR10}.
A set $X$ is a {\it variable length code} (a {\it code} for short) iff any equation among the words of $X$ is trivial, 
that is, for any pair of sequences of words in $X$, namely  $(x_i)_{1\le i\le m}$, $(y_j)_{1\le i\le n}$, the equation
$x_1\cdots x_m=y_1\cdots y_n$ implies  $m=n$ and $x_i=y_i$ for each integer $i\in [1,m]$.
By definition $X^*$, the submonoid of $A^*$ which is generated by $X$, is {\it free}. Equivalently, $X^*$ satisfies the property
of {\it{equidivisibility}}, that is 
$(X^*)^{-1}X^*\cap X^* (X^*)^{-1}=X^*$.

Some famous families of codes that have been studied in the literature:
$X$ is a {\it prefix } ({\it suffix, bifix}) {\it code} iff $X\neq\{\varepsilon\}$ and $X\cap XA^+=\emptyset$ $(X\cap A^+X=\emptyset,
X\cap XA^+=X\cap A^+X=\emptyset$).
$X$ is a code with a {\it  finite deciphering delay} ({\it f.d.d. code} for short) if it is a code
and if a non-negative integer $d$ exists such that $X^{-1}X^*\cap X^dA^+\subseteq X^+$. 
With this condition, if another integer $d'$ exists  such that we have $X^*X^{-1}\cap A^+X^{d'}\subseteq X^+$,
we say that $X$ is a {\it two-way f.d.d. code}.
$X$ is a {\it uniformly synchronized code} if it is a code and if a positive integer $k$ exists such that, 
for all $x,y\in X^k, u,v\in A^+$:
$uxyv\in X^*\Longrightarrow ux,xv\in X^*.$
 $X$ is a {\it circular code} if  for any pair of sequences of words in $X$, 
namely $(x_i)_{1\le i\le m},(y_j)_{1\le j\le n}$, and any pair of words $s,p$, with $s\neq\varepsilon$, the equation
$x_1\cdots x_m=sy_2\cdots y_np$, with $y_1=ps$, implies $m=n$, $p=\varepsilon$ and $x_i=y_i$ for each $i\in [1,m]$.

In the whole paper, we consider a {\it finite} or {\it countable} alphabet $A$ and a mapping $\theta$ which satisfies each of  the three following conditions:

(a) $\theta$ is  a one-to-one correspondence onto $A^*$

(b) $\theta$ is {\it literal}, that is $\theta(A)\subseteq A$

(c) either $\theta$ is a {\it morphism} or it is an {\it antimorphism} (it is an antimophism if $\theta(\varepsilon)=\varepsilon$ 
and  $\theta(xy)=\theta(y)\theta(x)$, for any pair of words $x,y$); for short in any case we write that $\theta$ is an {\it (anti)morphism}.

In the case where $A$ is a finite set, it is well known that a positive integer $n$ exists such that $\theta^n=id_{A^*}$.
In the whole paper, we are interested in the family of sets $X\subseteq A^*$ that are invariant under 
the mapping $\theta$ ({\it $\theta$-invariant} for short), that is $\theta(X)=X$. 
\section{A Defect Effect for Invariant Sets}
\label{Max}
Informally, the famous defect theorem says that if some words of a set $X$ satisfy a non-trivial equation, then these words may be written 
upon an alphabet of smaller size. 
In this section, we examine whether a corresponding result may be stated in the frameword of $\theta$-invariant sets.
The following property comes from the definition:
\begin{proposition}
\label{base}
Let $M$ be a submonoid of $A^*$ and let $S\subseteq A^*$ be such that $M=S^*$. 
Then $M$ is $\theta$-invariant if and only if $S$ is $\theta$-invariant.
\end{proposition}
Clearly the intersection of a non-empty family of $\theta$-invariant free submonoids of $A^*$ is itself a $\theta$-invariant free  submonoid. Given a submonoid $M$ of $A^*$, recall that its {\it minimal generating set} is $(M\setminus\{\varepsilon\})\setminus (M\setminus\{\varepsilon\})^2$.
\begin{theorem}
\label{defect}
Let $A$ be a finite or countable alphabet, let $X\subseteq A^*$ be a $\theta$-invariant set and let $Y$ be the minimal  generating set of the smallest $\theta$-invariant free submonoid of $A^*$ 
which contains $X$. If $X$ is not a code, then we have $|Y|\le |X|-1$.
\end{theorem}
\begin{proof}
With the notation of Theorem \ref{defect}, since $Y$ is  a code, each word $x\in X$ has a unique factorization upon the words of $Y$, 
namely $x=y_1\cdots y_n$, with $y_i\in Y$ ($1\le i\le n$). 
In a classical way, we say that $y_1$ ($y_n$)  is the {\it initial} ({\it terminal}) factor of $x$ (with respect to such a factorization). 
At first, we shall establish the following lemma:
\begin{lemma}
\label{initial}
With the preceding notation, each word in $Y$ is the initial (terminal) factor of a word in $X$.
\end{lemma}
\begin{proof}
By contradiction, assume that a word $y\in Y$  that is never initial of any word in $X$ exists. 
Set $Y_0=(Y\setminus \{y\})\{y\}^*$ and  $Y_i=\theta^i(Y_0)$, for each integer $i\in\Z$. 
In a classical way (cf e.g. \cite[p. 7]{Lo83}), since $Y$ is a code, $Y_0$ itself is a code.
Since $\theta^i$ is a one-to-one correspondence, for each integer $i\in\Z$, $Y_i$ is a code, that is $Y_i^*$ is a free submonoid of $A^*$. 
Consequently, the intersection, namely $M$, of the family $(Y_i^*)_{i\in\Z}$ is itself a free submonoid of $A^*$. 
Moreover we have  $\theta(M)\subseteq M$ (indeed,  given a word $w\in M$, $\theta(w)\not\in Y_i$ implies $w\not\in Y_{i-1}$)  therefore, since $\theta$ is onto, we obtain $\theta(M)=M$. 
Let $x$ be  an arbitrary word in  $X$. Since $X\subseteq Y^*$,  and according to  the definition of $y$, 
we have $x=(y_1y^{k_1})(y_2y^{k_2})\cdots (y_ny^{k_n})$, with $y_1,\cdots y_n\in Y\setminus\{ y\}$ and $k_1,\cdots k_n\ge 0$. 
Consequently $x$ belongs to $Y_0^*$, therefore we have $X\subseteq Y_0^*$. 
Since  $X$ is $\theta$-invariant, this implies $X=\theta(X)\subseteq Y_i^*$ for each $i\in\Z$, thus $X\subseteq M$.\\
But the word $y$ belongs to $Y^*$ and doesn't belong to $Y_0^*$  thus it doesn't belong to $M$.  
This implies $X\subseteq M\subsetneq Y^*$: a contradiction with the minimality of $Y^*$.\cqfd
\end{proof}
\begin{flushleft}
 {\it Proof of Theorem \ref{defect}}. Let $\alpha$ be the mapping from  $X$ onto $Y$ which, with every word $x\in X$, 
 associates  the initial factor of 
$x$ in its (unique) factorization over $Y^*$. According to Lemma \ref{initial}, $\alpha$ is onto.
We will prove that it is not one-to-one. Classically,  since $X$ is not a code, a non-trivial equation may be written among its words, say:
$x_1\cdots x_n=x'_1\cdots x'_m,~~ {\rm with}~~x_i,x'_j\in X~~x_1\neq x'_1~~(1\le i\le n, 1\le j \le m).$
Since $Y$ is a code, a unique sequence of words in $Y$, namely $y_1,\cdots, y_p$ exists such that:
$x_1\cdots x_n=x'_1\cdots x'_m=y_1\cdots y_p.$
 This implies $y_1=\alpha(x_1)=\alpha(x'_1)$ and completes the proof.
\cqfd
\end{flushleft}
\end{proof}
In what follows we discuss some  interpretation of Theorem \ref{defect} with regards to equations in words.
For this purpose, we assume that $A$ is finite, thus a positive integer $n$ exists such that $\theta^n=id_{A^*}$. 
Consider  a finite set of words, say $Z$, and denote by $X$ the union of the sets $\theta^i(Z)$, for $i\in [1,n]$; assume that a non-trivial equation holds among the words of $X$, namely
$x_1\cdots x_m=y_1\cdots y_p$.
By construction $X$ is $\theta$-invariant therefore, 
according to Theorem \ref{defect}, a $\theta$-invariant code $Y$ exists such that $X\subseteq Y^*$, with  $|Y|\le |X|-1$. 
This means that each of the words in $X$ can be expressed by making use of at most $|X|-1$ words of type $\theta^i(u)$, 
with $u\in Y$ and $1\le i\le n$.
It will be easily verified that the examples from \cite{CCKS11,KM08,MMNS14} corroborate this fact, moreover below  
we mention an original  one: 
\begin{example}
\label{ex1}
Let  $\theta$ be a literal antimorphism such that $\theta^3=id_{A^*}$.
Consider two different  words $x,y$, with $|x|>|y|$, which satisfy the following equation:
$$ x\theta(y)=\theta^2(y)\theta(x).$$
With these conditions, a pair of words $u,v$ exists such that $x=uv$, $\theta^2(y)=u$,
thus $y=\theta(u)$,
moreover we have  $v=\theta(v)$ and $u=\theta(u)=\theta^2(u)$.
With the preceding notation, we have $Z=\{x,y\}$, $X=Z\cup\theta(Z)\cup\theta^2(Z)$,
$Y=\{u\}\cup\{v\}\cup\{\theta(u)\}\cup\{\theta(v)\}\cup\{\theta^2(u)\}\cup\{\theta^2(v)\}$.
It follows from $y=\theta(y)=\theta^2(y)$ that $X=\{x\}\cup\{\theta(x)\}\cup\{\theta^2(x)\}\cup \{y\}$.\\
- At first, assume that no word $t$ exists such that $u,v\in t^+$. In a classical way,
we have $uv\neq vu$, thus $X=\{x,\theta(x),\theta^2(x),y\}$ and
$Y=\{u,v\}$.
We verify that  $|Y|\le |X|-1$.\\
- Now, assume that we have $u,v\in t^+$. We obtain $X=Z=\{x,y\}$ and $Y=\{t\}$. Once more we have   $|Y|\le |X|-1$.
\end{example}
\section{Maximal $\theta$-Invariant Codes}
\label{Maxcompl}
Given  set $X\subseteq A^*$, we say that it  is {\it thin} if $A^*\neq  F(X)$. 
Regular codes  are well known examples  of thin codes.  From the point of view of maximal codes, below we recall  one of the  famous result  stated by Sch\"utzenberger:
\begin{theorem}{\bf \cite[Theorem 2.5.16]{BPR10}}
\label{classic}
Let $X$ be an thin code.
Then the following conditions are equivalent:\\
(i) $X$ is complete\\
(ii) $X$ is a maximal code.
\end{theorem}
The aim of this section is to examine whether a corresponding result may be stated in the family of thin $\theta$-invariant codes.

In the case where $|A|=1$, we have  $\theta=id_{A^*}$, moreover the codes are all the singletons in $A^+$.  
Therefore any code is  $\theta$-invariant,  maximal and complete. In the rest of the paper, we assume that $|A|\ge 2$.
\begin{flushleft} 
{\it Some notations}.
Let $X$ be a non-complete $\theta$-invariant code, and let $y\not\in F(X^*)$. Without loss of generality, 
we may assume that the initial and the terminal letters of $y$ are different (otherwise, substitute to $y$ the word $ay{\overline a}$, 
with $a,{\overline a}\in A$ and $a\neq {\overline a}$), we may also assume that $|y|\ge 2$. 
Set:
\end{flushleft} 
\begin{eqnarray}
\label{Eq2}
y=ax{\overline a},~~ z={\overline a^{|y|}}ya^{|y|}={\overline a^{|y|}}ax{\overline a}a^{|y|}.
\end{eqnarray}
Since $\theta$ is a literal (anti)morphism, for each integer $i\in \Z$, a pair of different letters  $b,{\overline b}$ 
and a word $x'$ exist such that $|x'|=|x|=|y|-2$, and:
\begin{eqnarray}
\label{Eq3}
\theta^i(z)={\overline b^{|y|}}\theta^i(y)b^{|y|}={\overline b^{|y|}}bx'{\overline b}b^{|y|}.
\end{eqnarray}
Given two (not necessarily different) integers $i,j\in \Z$, we will accurately study how the two words $\theta^i(z), \theta^j(z)$ may overlap.
\begin{lemma}
\label{overlapZ}
With the notation in (\ref{Eq3}), let $u,v\in A^+$ and $i,j\in \Z$ such that  $|u| \le |z|-1$ and $\theta^i(z)v=u\theta^j(z)$. 
Then we have $|u|=|v|\ge 2|y|$, moreover a letter $b$ and a unique positive integer $k$ (depending of $|u|$) exist such that we have
 $\theta^i(z)=ub^k$, $\theta^j(z)=b^kv$, with $k\le |y|$. 
\end{lemma}
\begin{proof}
According to (\ref{Eq3}), we set 
$\theta^i(z)={\overline b^{|y|}}bx'{\overline b}b^{|y|}$
and 
$\theta^j(z)={\overline c^{|y|}}cx''{\overline c}c^{|y|}$,
with $b, {\overline b}, c, {\overline c}\in A$ and $b \neq {\overline b}, c \neq {\overline c}$. 
Since $\theta$ is a literal (anti)morphism, we have $|\theta^i(z)|=|\theta^j(z)|$ thus $|u|=|v|$; 
since we have $1\le|u|\le 3|y|-1$, exactly one of the following cases occurs:\\
{\it Case 1:} $1\le |u|\le |y|-1$. With this condition, we have 
($\theta^i(z))_{|u|+1}={\overline b}={\overline c}=(u\theta^j(z))_{|u|+1}$ and 
$(\theta^i(z))_{|y|+1}=b={\overline c}=(u\theta^j(z))_{|y|+1}$,
which contradicts $b\neq {\overline b}$.\\
{\it Case 2:} $|u|=|y|$. This condition implies 
$(\theta^i(z))_{|u|+1}=b={\overline c}=(u\theta^j(z))_{|u|+1}$ and 
$(\theta^i(z))_{2|y|}={\overline b}={\overline c}=(u\theta^j(z))_{2|y|}$, 
which contradicts $b\neq {\overline b}$.\\
{\it Case 3:} $|y|+1\le |u|\le 2|y|-1$. We obtain 
$(\theta^i(z))_{2|y|}={\overline b}={\overline c}=(u\theta^j(z))_{2|y|}$ and
$(\theta^i(z))_{2|y|+1}=b={\overline c}=(u\theta^j(z))_{2|y|+1}$
which contradicts $b\neq {\overline b}$.\\
{\it Case 4:} $2|y|\le|u|\leq 3|y|-1$. With this condition, necessarily we have  
$b={\overline c}$, therefore an integer $k\in [1,|y|]$ exists such that $\theta^i(z)=ub^k$ and $\theta^j(z)=b^kv$.
\cqfd
\end{proof}
Set $Z=\{\theta^i(z)|i\in \Z\}$.
Since $y \notin F(X^*)$ and since $X$ is $\theta$-invariant, 
for any integer $i\in\Z$ we have $\theta^i(z)\not\in F(X^*)$,
hence we obtain $Z\cap F(X^*)=\emptyset$.
By construction,  all the words in $Z$ have  length $|z|$
moreover, as a consequence of Lemma \ref{overlapZ}:
\begin{lemma}
\label{facteurInterne}
With the preceding notation, we have $A^+ZA^+\cap ZX^*Z=\emptyset$.
\end{lemma}
\begin{proof}
By contradiction, assume that $z_1,z_2,z_3\in Z$ , $x\in X^*$ and $u,v\in A^+$ exist such that $uz_1v=z_2xz_3$.
By comparing the lengths of the words $u$ and $v$ with $|z|$, exactly one of the three following cases occurs:\\
{\it Case 1:} $|z| \le |u|$ and $|z| \le|v|$. With this condition,  we have $z_2\in P(u)$ and $z_3\in S(v)$, therefore 
the word $z_1$ is a factor of $x$: this contradicts $Z\cap F(X^*)=\emptyset$.\\
{\it Case 2:} $|u|<|z|\le |v|$. We have in fact $u\in P(z_2)$ and $z_3\in S(v)$.
We are in the condition of Lemma  \ref{overlapZ}: the words $z_2$, $z_1$  overlap. 
Consequently, $u\in A^+$ and $b\in A$ exist such that $z_2=ub^k$ and $z_1=b^kz'_1$, with $1\le k\le |y|$. 
But by construction we have $|uz_1|=|z_2xz_3|-|v|$:
since we assume $|v|\ge |z|$,  this implies
 $|uz_1|\le|z_2xz_3|-|z|=|z_2x|$, therefore we  obtain $uz_1=ub^{k}z'_1\in P(z_2x)$. 
It follows from $z_2=ub^{k}$ that $z'_1\in P(x)$. 
Since $z_1\in Z$ and according to (\ref{Eq3}), $i\in\Z$ and ${\overline b}\in A$ exist such that we have $z_1=b^kz'_1=b^{|y|}\theta^i(y){\overline b}^{|y|}$.
Since by Lemma  \ref{overlapZ} we have $|z_1'|=|u| \geq 2|y|$, we obtain
$\theta^i(y) \in F(z'_1)$, which contradicts $y \notin F(X^*)$.\\
{\it Case 3:} $|v|<|z|\le |u|$. Same arguments on the reversed words lead to a conclusion similar to that of  Case 2.\\
{\it Case 4:} $|z|>|u|$ and $|z|>|v|$. With this condition, both the pairs of words $z_2,z_1$ and $z_1,z_3$ overlap. 
Once more we are in the condition of Lemma \ref{overlapZ}: letters $c,d$, 
words $u, v, s,t$, and integers $h,k$ exist such that the two following properties hold:
\begin{eqnarray}
z_2=uc^h,~~  z_1=c^{h}s,~~ |u|=|s| \ge 2|y|, ~~h \le |y|,
\\
z_1=td^{k},~~ z_3=d^kv, ~~ |v|=|t| \ge 2|y|, ~~k \le |y|.
\end{eqnarray}
It follows from $uz_1v=z_2xz_3$ that $uz_1v=(uc^h)x(d^kv)$, thus  $z_1=c^{h}xd^{k}$.
Once more according to (\ref{Eq3}), $i\in\Z$ and ${\overline c}\in A$ exist such that we have $z_1=c^{|y|}\theta^i(y){\overline c}^{|y|}$. 
Since we have $h,k\le |y|$, this implies $d={\overline c}$
moreover $\theta^i(y)$ is a factor of $x$. 
Once more, this contradicts $y \notin F(X^*)$.
\cqfd
\end{proof}
\begin{figure}
\label{f1}
\begin{center}
\includegraphics[width=9.5cm,height=3cm]{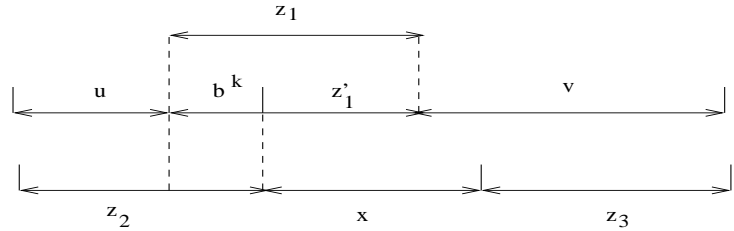}
\end{center}
\caption[]
{Proof of  Lemma \ref{facteurInterne}: Case 2}
\end{figure}
Thanks to Lemma  \ref{facteurInterne} we will prove some meaningful  results in Section \ref{Max-family}. Presently, 
we will apply it in a special context:
\begin{corollary}
\label{X-pref}
With the preceding notation, $X^*Z$ is a prefix code.
\end{corollary}
\begin{proof}Let $z_1,z_2\in Z$, $x_1,x_2\in X^*, u \in A^+$, such that $x_1z_1u= x_2z_2$. For any word $z_3\in Z$, 
we have  $(z_3x_1)z_1(u)=z_3x_2z_1$, a contradiction with Lemma \ref{facteurInterne}.\cqfd
\end{proof}
We are now ready to prove the main result of the section:
\begin{theorem}
\label{equivalences}
Let $A$ be a finite or countable alphabet and let $X\subseteq A^*$ be a thin $\theta$-invariant code.
Then the following conditions are equivalent:\\
(i) $X$ is complete\\
(ii) $X$ is a maximal code\\
(iii) $X$ is maximal in the family  $\theta$-invariant codes.
\end{theorem}
\begin{proof}
Let $X$ be  a $\theta$-invariant code. According to Theorem \ref{classic}, if $X$ is thin and complete, then it is a maximal code,
therefore $X$ is maximal in the family of $\theta$-invariant codes. 
For proving the converse, we consider a set $X$ which is maximal in the family of 
$\theta$-invariant codes.\\
Assume that $X$ is not complete and let $y\not\in F(X^*)$.
Define the word $z$ as in (\ref{Eq2}) and consider the set $Z=\{\theta^i(z) | i \in \Z\}$. 
At first, we will prove that $X\cup Z$ remains a code. In view of that, we consider an arbitrary equation between the words in $X\cup Z$. Since $X$ is a code, without loss of generality, we may assume that at least one element of $Z$ has at least one occurrence in one of the two sides of this equation. As a matter of fact, with such a  condition and since $Z\cap F(X^*)=\emptyset$, two sequences of  words in $X^*$, namely $(x_i)_{1\le i\le n},(x'_j)_{1\le j\le p}$ and two sequences of words in  $Z$, namely 
$(z_i)_{1\le i\le n-1}, (z'_j)_{1\le j \le p-1}$ exist such that the equation takes the following form:
\begin{equation}\label{Eq1} x_1z_1x_2z_2\cdots x_{n-1}z_{n-1} x_n=x'_1z'_1x'_2z'_1\cdots x'_{p-1} z'_{p-1}x'_p.\end{equation}
Without loss of generality, we assume $n\ge p$. At first,  according to Corollary  \ref{X-pref}, necessarily, we have $x_1=x'_1$, therefore Equation (\ref{Eq1}) is equivalent to:
$z_1x_2z_2\cdots x_{n-1}z_{n-1} x_n=z'_1x'_2z'_2\cdots x'_{p-1} z'_{p-1}x'_p,$
however, since all the words in $Z$ have a common length, we have $z_1=z'_1$ hence our equation is equivalent to $x_2z_2\cdots x_{n-1}z_{n-1} x_n=x'_2z'_2\cdots x'_{p-1} z'_{p-1}x'_p.$
Consequently, by  applying iteratively  the result of Corollary \ref{X-pref}, we obtain:
 $x_2=x'_2,\cdots, x_p=x'_p$,
which implies $x_{p+1}z_{p+1}\cdots z_{n-1}x_n=\varepsilon$, thus $n=p$.  In other words Equation (\ref{Eq1}) is trivial, thus $X\cup Z$  is a code.\\
Next, since $\theta$ is one-to-one and since we have  $\theta(X\cup Z)\subseteq \theta(X)\cup\theta(Z)=X\cup Z$, the code $X\cup Z$ is $\theta$-invariant. It follows from $z\in Z\setminus X$ that  $X$ is strictly included in $X\cup Z$: this contradicts the maximality of $X$ in the whole family of $\theta$-invariant codes, and completes the proof of Theorem \ref{equivalences}.\cqfd
\end{proof}
\begin{example}
Let $A=\{a,b,c\}$. Consider the antimorphism $\theta$ which is generated by the permutation $\sigma(a)=b,\sigma(b)=c,\sigma(c)=a$ and let
 $X=\{ab,cb, ca, ba, bc, ac\}$; it can be easily verified that $X$ is a $\theta$-invariant code.
Since we have $a^3\not\in F(X^*)$, by setting  $y=a^3b$ and $z=b^4\cdot a^3b\cdot a^4$ we are in Condition (\ref{Eq2}).
The corresponding set $Z$ is $\{\theta^i(z)|i\in\Z\}$
 $=\{b^4cb^3c^4, a^4c^3ac^4, a^4ba^3b^4, c^4b^3cb^4, c^4ac^3a^4, b^4a^3ba^4\}$.
Since $X\cup Z$ is a prefix set, this guarantees that $X\cup Z$ remains a $\theta$-invariant code.
\end{example}
\section{Maximality in Some Families of $\theta$-Invariant Codes}
\label{Max-family}
In the literature,  statements similar to Theorem \ref{classic} were established in the framework of  some  special families of thin codes. In this section we will draw similar investigations with regards to $\theta$-invariant codes.
We will establish the following result:
\begin{theorem}
\label{equivalences-family}
Let $A$ be a finite or countable alphabet and let $X\subseteq A^*$ be a thin $\theta$-invariant prefix (resp. bifix, f.d.d., two-way f.d.d, uniformly synchronized, circular) code.
Then the following conditions are equivalent:\\
(i) $X$ is complete\\
(ii) $X$ is a maximal code\\
(iii)  $X$ is maximal in the family of prefix (bifix, f.d.d., two-way f.d.d, uniformly synchronized, circular) codes\\
(iv) $X$ is maximal in the family  $\theta$-invariant codes\\
(v) $X$ is maximal in the family of $\theta$ invariant prefix (bifix, f.d.d., two-way f.d.d, uniformly synchronized, circular) codes.
\end{theorem}
\skpr
According to to Theorem \ref{equivalences}, and thanks to \cite[Proposition 3.3.8]{BPR10}, \cite[Proposition 6.2.1]{BPR10}, \cite[Theorem 5.2.2]{BPR10},  \cite[Proposition 3.6]{Br98} and  \cite[Theorem 3.5]{N08}, if $X$ is complete then it is maximal in the family  of $\theta$-invariant codes and maximal in the family of $\theta$-invariant prefix (bifix, f.d.d., two-way f.d.d, uniformly synchronized, circular) codes.
Consequently, the proof of Proposition \ref{equivalences-family} comes down to 
establish that if $X$ is not complete, then it cannot be maximal in the family of $\theta$-invariant prefix (bifix, f.d.d., wo-way f.d.d, uniformly synchronized, circular) codes.

1) We begin by $\theta$-invariant prefix codes.
At first, we assume that $\theta$ is an antimorphism.
Since $X\cap XA^+=\emptyset$, and since $\theta$ is injective, we have $\theta(X)\cap\theta(XA^+)=\emptyset$, thus $X\cap A^+X=\emptyset$, hence $X$ is also a suffix code.
Assume that  $X$ is   not complete. According to \cite[Proposition 3.3.8]{BPR10},  it is non-maximal in both the families of prefix codes and suffix codes. Therefore a pair of words $y,y'\in A^+\setminus X$ exists such $X\cup\{y\}$ ($X\cup\{y'\}$) remains a prefix (suffix) code. 
By construction $X\cup\{yy'\}$ remains a  code which is both prefix and suffix.\\
Set $Y=\{\theta^i(yy')| i\in \Z\}$: since all the words in $Y$ have same positive length, $Y$ is a prefix code.
From the fact that $\theta$ is one-to-one, for any integer $i\in \Z$ we obtain
$\theta^i(\{yy'\})\cap \theta^i(P(X))=\theta^i(X)\cap P( \theta^i(yy'))=\emptyset$, consequently
 $X\cup Y$ remains a prefix code. By construction, $Y$ is $\theta$-invariant and it is not included in $X$, thus $X$ is not a maximal prefix code.\\
In the case where $\theta$ is a morphism, the preceding arguments may be simplified. Actually, a word $y\in A^+\setminus X$ exists such that $X\cup \{y\}$ remains a prefix code, thereferore by setting $Y=\{\theta^i(y)| i\in \Z\}$, $X\cup Y$ remains a prefix code.

2) (sketch) The preceding arguments may be applied for proving  that in any case, if $X$ is a non-complete bifix code, then it is maximal.

3,4) (sketch) In the case where $X$ is a (two-way)  f.d.d.-code, according to \cite[Proposition 5.2.1]{BPR10}, similar arguments leads to a similar conclusion.

5) In the case where $X$ is a $\theta$-invariant uniformly synchronized code with {\it verbal delay} $k$ (\cite[Section 10.2]{BPR10}), we must make use of different arguments. Actually, according to \cite[Theorem 3.10]{Br98}, a complete uniformly synchronized code $X'$ exists, with synchronizing delay $k$,  and such that $X\subsetneq X'$.
More precisely, $X'$ is the minimal generating set of the submonoid $M$ of $A^*$ which is defined by
 $M=(X^{2k}A^*\cap A^* X^{2k})\cup X^*$.
According to Proposition \ref{base} in the present paper, $X'$ is $\theta$-invariant.  Since $X$ is stictly included in $X'$, it cannot be maximal in the family of $\theta$-invariant uniformly synchronized codes with delay $k$. 

6) It remains to study the case where $X$ is a non-complete $\theta$-invariant circular code. Let $y\not\in F(X^*)$ and let $z$ and  $Z$ be computed as in Section \ref{Max}: this  guarantees that $X\cup Z$ is a  $\theta$-invariant set.
For proving that $X\cup Z$ is  a circular code,
by contradiction  we assume that some words $y_1, \cdots y_n, y'_1,\cdots, y'_m\in X\cup Z$ (with $m+n$ minimal), $p\in A^*$ and $s\in A^+$, exist  such that the following equation holds:
\begin{eqnarray}
\label{circulaire}
y_1y_2\cdots y_n=sy'_2y'_3\cdots y'_mp~~~{\rm and}~~~y'_1=ps.
\end{eqnarray}
Once more since $X$ is a code, and since $Z\cap F(X^*)=\emptyset$, without loss of generality we assume that at least one integer $i\in\Z$ exists such that $y_i\in Z$; similarly, at least one integer $j\in [1,m]$ exists such that $y'_j\in Z$.
By construction, we have $y_i\in F(y'_j \cdots y'_my'_1\cdots y'_j \cdots y'_my'_1 \cdots y'_j)$; consequently, since all the words in $Z$ have the same length, a pair of integers $h,k\in [1,m]$  and
 a pair of words $u,v$ exist such that 
$uy_iv\in y'_hX^*y'_k$. According to Lemma \ref{facteurInterne}, necessarily we have
either $u=\varepsilon$ or $v=\varepsilon$;  this implies $y_i=y'_h$ or $y_i=y'_k$, which
contradicts the minimality of $m+n$, therefore $X\cup Z$ is a circular code.
\cqfd
\section{Embedding a Regular Invariant Code into a Complete One}
In this section, we consider a non-complete regular $\theta$-invariant code $X$ and we are interested in the problem of computing a  complete one, namely $Y$, such that $X\subseteq Y$. 
Historically, such a question appears for the first time in \cite{R75}, where the author asked for the possibility of embedding a finite code into a regular complete one. 
With regards to $\theta$-invariant codes, it seems natural to generalize the formula from
\cite{ER85} by making use of the code $Z$ that was introduced in Section \ref{Maxcompl}.
More precisely we would consider the set $X'=X\cup (ZU)^*Z$, with $U=A^*\setminus (X^*\cup A^*ZA^*)$. Unfortunately, with such a construction we observe that some pairs of words in $Z$ may overlap, therefore a non-trivial equation could hold among the words of $X'$.

Nevertheless, we shall see that in the very special case where $\theta$ is an involutive antimorphism, convenient invariant overlapping-free words can be computed. Denote by $\theta_0$   the antimorphism which is generated by the identity onto $A$; in other words, with every word $w=w_1\cdots w_n\in A^*$ (with $w_i\in A$, for $1\leq i\le n$), it associates  $\theta_0(w)=w_n\cdots w_1$.

\begin{proposition}
\label{completion}
Let $A$ be a finite alphabet and let $\theta$ be an antimorphism onto $A^*$, with $\theta\neq\theta_0$.
If $\theta$ is involutive, then any non-complete regular  $\theta$-invariant code can be embedded into a complete one.
\end{proposition}
%
\begin{proof}
Let $X$ be such that $\theta(X)=X$. Assume that $X$ is not complete.
 We will construct an overlapping-free word $t\notin F(X^*)$ such that $\theta(t)=t$.
At first, we consider a word $x$ such that $x\not\in F(X^*)$ and $|x| \ge 2$. 
Without loss of generality, we assume that $x$ is overlapping-free (otherwise,  as in  \cite[Proposition 1.3.6]{BPR10}, a word $s$ exists such that $xs$ is overlapping-free).
If $\theta(x)=x$, then we set $t=x$,
otherwise let $y=cx$, where $c$ stands for the initial letter of $x$.
Once more, without loss of generality we assume that $y$ is overlapping-free.
By construction we have $y\in ccA^+$, thus $|y|\ge3$ and $y_1=y_2=c$. 
If $\theta(y)=y$, then we set $t=y$. Now assume $\theta(y)\neq y$; according to the condition of Proposition \ref{completion}, we have $\theta|_A \neq id_A$,  therefore a pair of  letters $a,b$ exists such that the following property holds:
\begin{eqnarray}
\label{Condition-abc}
a\neq b,~~~b\neq c,~~~\theta(a)=b,~~~\theta(b)=a.
\end{eqnarray}
Set   $t=a^{|y|}b\theta(y) y a b^{|y|}$.  By construction, we have $\theta(t)=t$, moreover the following property holds:
{\bf \begin{claim}
\label{no-overlap}
{\rm$t$ is an overlapping-free word.}
\end{claim}}
\begin{proof}
Let $u,v\in A^*$ such that $ut=tv$, with $1 \leq |u|\le |t|-1$. 
According to the length of $u$, exactly one of the following cases occurs:\\
{\it Case 1:} $1\le |u|\le |y|$. With this condition, we obtain 
$t_{|y|+1}=b=(ut)_{|y|+1}=a$: a contradiction with $a\neq b$.\\
{\it Case 2:} $|y|+1\le |u|\le 2|y|$. This condition implies $\theta(y_1)=t_{2|y|+1}=a$, therefore we obtain $c=y_1=\theta(a)=b$: a contradiction with (\ref{Condition-abc}).\\
{\it Case 3:} $|u|=2|y|+1$. We have $y=a^{|y|}$: since we have  $|y|\ge 3$, this contradicts the fact that $y$ is overlaping-free.\\
{\it Case 4:}  $|u|=2|y|+2$. We have  
$t_{2|y|+3}=y_2=c=(ut)_{2|y|+3}=a$. It follows from $y_1=y_2=c$ that $y=a^{|y|}$: once more this contradicts the fact that $y$ is overlapping-free.\\
{\it Case 5:} $2|y|+3\le |u|\le 3|y|+2$. 
By construction, we have $t_{|u|+|y|}=b=(ut)_{|uy|}=a$, a contradiction with 
(\ref{Condition-abc}).\\
{\it Case 6:} $3|y|+3\le |u|\le |t|-1=4|y|+1$.  
We obtain $t_{|u|+1}=b=(ut)_{|u|+1}=a$: once more this contradicts (\ref{Condition-abc}). \\
In any case we obtain a contradiction: this establishes the claim.
\end{proof}
Since we have $t\not\in F(X^*)$, and since $t$ is overlapping-free, the classical method from \cite{ER85} may be  applied without any modification to ensure that $X$ may embedded into a complete code, say $X'$. Recall that it computes in fact  a code $X'$ as $X\cup V$, with
$V=t(Ut)^*$ and $U=A^*\setminus (X^*\cup A^*tA^*)$.
Moreover, since $\theta(t)=t$, it is straightforward to verify that $\theta(X')=X'$.\cqfd
\end{proof}
\begin{figure}
\label{f2}
\begin{center}
\includegraphics[width=10.5cm,height=3.5cm]{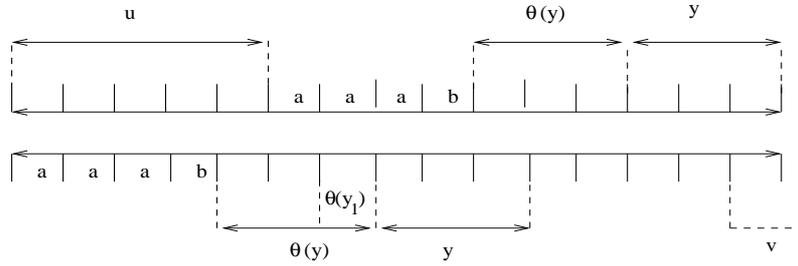}
\end{center}
\caption[]
{Proof of Proposition \ref{completion}: Case 2 with $|y|=3$ and $|u|=5$}
\end{figure}
With regards to the antimorphism $\theta_0$, 
necessarily the words $w$, $\theta_0(w)$ overlap, therefore the preceding methodology seems to be unreliable  in the most general case. We finish our paper by stating the following open problem:
\begin{flushleft}
{\it Problem.}
Let $A$ be a finite alphabet and let $\theta$ be an (anti)morphism onto $A^*$. Given a non-complete regular $\theta$-invariant code $X\subset A^*$,
can we compute  a complete  regular $\theta$-invariant code $Y$ such that $X\subseteq Y$?
\end{flushleft}

\bibliography{mysmallbib}{}
\bibliographystyle{plain}

\end{document}